%% file: MAIN.tex
\documentclass[11pt]{article}

\usepackage{amsmath,amssymb,amsthm,amsfonts,latexsym,bbm,xspace,graphicx,float,mathtools}
\usepackage[letterpaper,margin=1in]{geometry}
\usepackage{color}
\usepackage{comment}
\usepackage{algorithmic,algorithm}

\usepackage{fullpage}
\usepackage{appendix}
\usepackage[nomessages]{fp}

\newtheorem{theorem}{Theorem}

\newtheorem{corollary}{Corollary}
\newtheorem{lemma}{Lemma}

\newtheorem{proposition}{Proposition}

\newtheorem{definition}{Definition}

\newtheorem{remark}{Remark}

\newcommand{\junk}[1]{}

\newcommand{\james}[1]{
\noindent{\textcolor{blue}{\textbf{ JZ:} \textsf{#1}}}
}

\junk{

}

\newcommand{\notshow}[1]{{}}

\definecolor{MyGray}{rgb}{0.8,0.8,0.8}

\begin{document}
\title{Signal to noise in matching markets}
\author{S. Matthew Weinberg\thanks{Computer Science, Princeton University. E-mail: smweinberg@princeton.edu. Work done in part while the author was a research fellow at the Simons Institute for the Theory of Computing.} \and James Zou\thanks{Stanford University. E-mail: jamesz@stanford.edu.}}
\addtocounter{page}{-1}
\maketitle
\begin{abstract} 

In many matching markets, one side ``applies'' to the other, and these applications are often expensive and time-consuming (e.g. students applying to college). It is tempting to think that making the application process easier should benefit both sides of the market. After all, the applicants could submit more applications, and the recipients would have more applicants to choose from. In this paper, we propose and analyze a simple model to understand settings where \emph{both sides} of the market suffer from increased number of applications.


The main insights of the paper are derived from quantifying the signal to noise tradeoffs in random matchings, as applications provide a signal of the applicants' preferences. When applications are costly the signal is stronger, as the act of making an application itself is meaningful. Therefore more applications may yield potentially better matches, but fewer applications create stronger signals for the receiving side to learn true preferences. 

We derive analytic characterizations of the expected quality of stable matchings in a simple utility model where applicants have an overall quality, but also synergy with specific possible partners. Our results show how reducing application cost without introducing an effective signaling mechanism might lead to inefficiencies for both sides of the market.


\end{abstract}

\newpage
\thispagestyle{empty}
\input{intro}
\input{model}

\input{results}

\input{applications}
\input{equiv_proof}

\bibliographystyle{alpha}
\bibliography{MasterBib} 

\input{appendix}

\end{document}

%% file: intro.tex

\section{Introduction}

College application is an expensive and time-consuming process. A student may apply to many colleges, and each application could take considerable effort to complete. Because of this, systems that reduce the cost of application, for example by allowing a common application to be submitted to many schools for free, would appear to benefit students. This may appear like a win-win because it seems  that colleges also stand to gain from receiving more applications. After all, having a larger pool of applicants to select from should make it more likely to find the students that are good fits. Many US universities now explicitly try to increase the number of applications. At the same time, lower acceptance rates are advertised as a badge of exclusivity, suggestive of higher quality student body. 

In this paper, we present and analyze a simple model that illustrates how when it becomes too easy to apply to more colleges, \emph{everyone}---students and colleges---might suffer. Students get into schools that are lower on their preference list and schools enroll students that are poor fits. The heart of the issue is signal to noise. Matching students with colleges is a learning problem. From an application, a college derives a noisy signal about how good of a fit is the student to that particular environment, and admission offers are made based on these noisy signals. Increasing the number of applications could increase the amount of noise in the entire matching system, making the learning more difficult and leading to a worse outcome for students and colleges. After all, when there is an cost (or at least an opportunity cost) to applying, the application itself carries a meaningful signal. We derive an analytic characterization of how the signal-to-noise trades off with the number of applications. While we motivate our analysis with the college-application example, the insights abouts signal and noise in multi-agent matchings is more general and can be applied in other settings. 

Since college admissions approximately follow a university-proposing deferred-acceptance procedure,\footnote{That is, colleges admit their favorite applicants, and students accept their favorite ``proposal.'' University ``wait-lists'' make the procedure even closer to a true deferred-acceptance protocol.} it is generally well-understood that students experience a tradeoff as the number of applications grows. Too few applications results in many students not being admitted at all, but too many applications results in students getting matched to universities further down their list (as the resulting admissions become more in line with the schools' preferences and less so with students')~\cite{Pittel89,KnuthMP90}. So already from existing theory on matching markets, one could conclude that the common app (i.e. a framework that enables a student to apply to many schools essentially for free) is not unilaterally beneficial for students, and some care is required to ensure that the application rate lies in a profitable range.

On the other hand, if schools can fully learn their preferences over students from reading applications, it is similarly well-understood that schools are more likely both to be matched at all, and to be matched with more preferable students as the application rate increases. However, there's much more to a student than can be learned simply through their application, and students often have some latent factors that affect how well they fit a particular college. For instance, maybe a student fits in with the culture at a specific school (and is therefore more likely to succeed than a comparable student who doesn't), or has a significant other already attending (and is therefore more likely to succeed without the additional travel), etc.\footnote{In these examples, we are making the mild assumption that school preferences are at least positively correlated with the likelihood of student success.} Students have a pretty good idea of whether or not they synergize with a university, but cannot genuinely convey this to universities through an application - clever students will simply claim synergy on all their applications. So without a formal signaling scheme (such as that used by the Economics academic job market - more on that in Section~\ref{sec:related}), the only signal the universities receive about possible synergy comes from the students' choices of where to apply.

In light of this, one might reasonably expect that the common app is not even unilaterally beneficial for universities, as they too experience a tradeoff: More applications results in potentially preferable admits, but also in a weaker signal regarding synergy. The main contribution of this paper is a simple model illustrating this tradeoff, and a clean analysis of the optimal application rate as a function of modeling parameters. Concrete details on the model appear in Section~\ref{sec:prelim}, and we briefly overview the model and our results below.

\subsection{Overview of Results}
To map our model to the example of college applications, it is better to think of restricting attention to the top-tier students and top-tier universities. Within this set of universities, one could imagine that preferences over universities are roughly symmetric (i.e. independently, uniformly at random for all students). We model there being a synergy between every student and their favorite college (so this is already captured in the student preferences). When colleges receive applications, applications from students with synergy may appear stronger on average than those without, but otherwise the preferences based on applications are similarly symmetric and independent across schools (formal definition in Section~\ref{sec:prelim}). Based on these noisy signals, each college offers admissions to its ``strongest'' applicants. Students with multiple offers selects their favorite to accept, (perhaps) additional admission offers are made by colleges and the process iterates.

Aside from postponing details about how signal strength corresponds to synergy (which is a modeling parameter), this is the entire model. We show in Section~\ref{sec:results} how to compute, for any application rate $K$ (number of applications per student), the expected number of schools who get their first, second, etc. choice student, as well as the expected number of schools who accept a student with synergy. We show that the number of students getting their first, second, etc. choice school must satisfy a system of (non-linear) equations. We derive the system of equations by considering a simple thought experiment, detailed in Section~\ref{sec:results}. 

In Section~\ref{sec:applications}, we consider some natural instantiations of our modeling parameters and derive the optimal application rate. The key takeaway from both sections together is that our simple model, which only concerns preferences and synergy, and not additional features like the cost of processing applications, is already rich enough to exhibit a tradeoff in utilities for both schools and students as the number of applications grows. 

While we use the language of students and universities throughout the paper, this should be viewed as a motivating example. Our analyses shed some light on the college application process (i.e. one should be careful about making applications too easy without a formal signaling scheme), but our model is intentionally too simple to make specific claims about the current practice (i.e. of the form ``the existence of the common app helps/hurts the average student/university''). We consider our main contribution to be a significantly simpler model than in prior literature that reaches the same qualitative conclusions (more on that below). 

\subsection{Related Work}\label{sec:related}
Following numerous seminal papers on stable matchings, e.g.~\cite{GaleS62, McVitieW71, GaleS85, Roth82}, much is known about algorithms for finding stable matchings, when matchings tend to favor one side or the other, and when they can be manipulated. We will not survey the field in its entirety, but discuss the most related works on three directions below.

\paragraph{Stable Matchings with Random Preferences.} Early work of~\cite{Pittel89,KnuthMP90} established a large gap between the utility of the proposing side and the receiving side for the deferred-acceptance algorithm in perfectly random markets with $n$ parties on each side of capacity $1$: the proposing side gets a match on average equal to their $\ln(n)^{th}$ favorite partner, whereas the receiving side on average gets their $(n/\ln n)^{th}$ partner (interestingly, it was recently shown that adding a single party to the proposing side causes the values to swap and the receiving side becomes better off~\cite{AshlagiKL16}). Recent work has also used random preferences as a lens to study various tie-breaking rules in school choice~\cite{Arnosti16, AshlagiNR15,AshlagiS14,AshlagiN16}. The works with techniques most related to ours are~\cite{ImmorlicaM15,KojimaP09}, who study random markets where one side proposes but to a vanishing fraction of the receiving side. The focus of these works is on understanding when such markets are manipulable.

While our techniques are indeed similar to those used in these works, to the best of our knowledge our approach for computing the quality profile of matches via systems of non-linear equations is novel, and could be of independent interest for related work along these lines.

\paragraph{Stable Matchings with Truncated Preference Lists.} In addition to the already discussed works of~\cite{ImmorlicaM15,KojimaP09}, truncated preference lists (e.g. where participants propose to/accept proposals from a proper subset of potential partners) have also been extensively studied outside of random markets~\cite{GusfieldI89, KobayashiM09, GonczarowskiF13, PiniRVW11}, mostly with an emphasis on how protocols can be manipulated via truncating preference lists. In this work we don't address strategic manipulations, and only study truncation in random markets.

\paragraph{Signaling in Matching Markets.} There is already much literature concerning matching markets with information asymmetries. Numerous works, both theoretical and empirical, show that the ability to signal interest in a potential partner can help improve match quality as long as there is at least some opportunity cost for signaling (e.g. only $K$ signals can be sent)~\cite{HalaburdaPPY16,RothX97,CheK15, LeeNKK11,LeeN15,ColesCLNRS10}. In particular,~\cite{ColesCLNRS10} discusses a signaling scheme used on the Economics academic job market: candidates are allowed up to two signals to send along with an application to alert the recipient to their interest if it might not be clear from the application alone (e.g. if an exceptionally strong candidate has family near a lower-tier school and might prefer to go there over higher ranked schools. Apparently it is not common practice to signal top-tier schools, as they are already aware of every applicant's high interest). 

One can interpret sending an application itself as a form of signaling, but in this work we do not discuss formal signaling mechanisms outside of the application process itself. It is worth noting that well-designed signaling schemes used in practice do have a positive impact on match quality~\cite{LeeNKK11,LeeN15,ColesCLNRS10}, even when signals are free but limited. So there is certainly cause to believe that understanding the tradeoff between application rate and match quality could benefit markets without formal signaling schemes.\footnote{Note that the common app does have some form of signaling scheme via early decision/early action.}

\paragraph{Tradeoffs as Application Rate Increases.} Somewhat surprisingly, there is not much work on understanding the tradeoffs experienced by the receiving side as application rates increase. To the best of our knowledge, this was first studied in recent work of~\cite{ArnostiJK15}, who show that when it is costly for the receiving side to evaluate applications and users on both sides arrive and depart, that the receiving side suffers under excessive applications. While their model is more realistic and closer to applications like oDesk, our model is comparatively simpler and cleaner for theoretical study. For instance, we identify a tradeoff for the receiving side as application rate increases solely due to information asymmetry, without appealing to the cost of evaluating applications. Clearly our simple model is less apt for practice, but perhaps more apt for tractable theory and intuition. 

\paragraph{The Effect of Signaling on College Admissions.} A recent article in Washington Monthly discussed issues with the college application process due to increased number of application~\cite{Kim16}. For example, schools have difficulty in predicting how many accepted students will choose to enroll, sometimes leading to prohibitive levels of under/over-enrollment (instances of miscalculation both ways have been severe enough to cause presidents' resignations). To cope with this, universities are trying to use ad hoc signalling measures such as early decision/early action, or ``did this applicant come for a campus visit?'' While there's nothing inherently wrong with this, both metrics are believed to benefit wealthier applicants who can make an early decision without competing financial aid packages, or afford campus visits. Addressing such complex issues requires a solid theoretical understanding of the signal-to-noise tradeoff in the application process.

%% file: model.tex

\section{Model and Preliminaries}\label{sec:prelim}

\subsection{Our Model}\label{sec:model}

\textbf{Students and Universities.} There is a set $S$ of students and $U$ of universities. Each student has a preference ordering over $U \cup \{\bot\}$, and each university has a preference ordering over $S \cup \{\bot\}$ ($\bot$ represents no match). We denote by $n = |S|$ the number of students, and $M = |U|/|S|$ the ratio of universities to students. We will treat $M$ as a constant and let $n$ grow. One can simulate a student $s$ not applying to a university $u$ by having $s$ rank $\bot$ above $u$. 

\textbf{Capacities.} Each student will accept admission to at most one university. Each university can enroll at most $L$ students. We will abuse notation and refer to a \emph{matching} as any mapping from students to universities where each student is matched with at most one university, and each university is matched with at most $L$ students.

\textbf{Student Preferences.} Students' preferences are independent and chosen uniformly at random over the space of all university orderings. If a student's favorite university is $u$, we say she is a \emph{special-$u$} student.


\textbf{University Preferences.} When we refer to university preferences, we mean their preferences only having seen applications, which are not necessarily perfectly correlated with the true utility to a university for accepting a certain student. University preferences are drawn in the following way: There are two distributions $D$ and $D'$ over $\mathbb{R}_+$. For university $u$, and each special-$u$ student, the university observes a signal drawn from $D$. For all other students, the university observes a signal drawn from $D'$. Then, the university ranks all students in decreasing order of signal. If $D = D'$, then the university's preferences are also uniformly random. 

\textbf{Utilities.} We do not explicitly model student utilities, other than that students get at least as much utility for going to more preferred schools. Similarly, we do not explicitly model university utilities other than that they get higher base utility for successfully recruiting more preferred students, and bonus utility for recruiting a student with synergy. All of our analysis directly studies the expected quality of matches (in terms of where partners rank on respective preference lists, and whether or not there is synergy), and can therefore be readily applied to any utility model with the preceding properties.


\textbf{Applications and Admissions.} Every student sends applications to their $K$ favorite schools. We do not model how it comes to be that this is the case (perhaps it is hard-coded into a centralized system, perhaps application costs/burdens are tuned to make this an equilibrium, etc.). Once all applications are submitted, the school-proposing deferred-acceptance procedure begins: school admit their $L$ favorite students and waitlist the rest (among the applications that it received). Students deny all but their favorite offer. Universities whose offers were denied make new offers to their next favorite students, and this process repeats until every university either has $L$ admitted students or no more offers to make. 

\textbf{Goal.} Our goal is to understand the expected match quality for the average student and university for varying settings of parameters as $n \rightarrow \infty$, while $K, L, M$ remain constant.  

\subsection{Stable Matching Preliminaries}\label{sec:prelimstable}
In this section, we review some basic preliminaries from stable matching. 

\textbf{Stable Matchings.} A matching $\mathcal{M}$ from $S$ to $U$ (where perhaps not everyone is matched) is \emph{stable} if it has no \emph{blocking pairs}. A blocking pair is a student $s$ and a university $u$ such that $s$ prefers $u$ to their partner in $\mathcal{M}$, and $u$ prefers $s$ to one of their partners in $\mathcal{M}$ (if $s$ is unmatched, their partner is $\bot$. If $u$ is matched to $\ell < L$ students, the remaining $L-\ell$ partners are $\bot$). The idea is that $\mathcal{M}$ is not stable because $s$ and $u$ prefer to break their matches in $\mathcal{M}$ and partner with each other instead. We do not consider university preferences over sets of students. 

\textbf{Finding Stable Matchings.} Seminal work of Gale and Shapley~\cite{GaleS62} shows that stable matchings always exist and can be found efficiently via the deferred acceptance algorithm: Initially, all universities are unmatched, and all students are matched to $\bot$. One at a time, an arbitrary unfilled university $u$ is selected to propose (unfilled universities are those with fewer than $L$ matches). The university proposes to their favorite student who has not yet rejected them. If that student is $\bot$, the proposal is accepted and $u$ becomes matched to $\bot$. Otherwise, the recipient, $s$, compares $u$ to their current match, $u'$, and chooses the university they like best. If it is $u$, $s$ is now matched to $u$ and $u'$ becomes unfilled. Otherwise, $u$ stays unfilled (and $s$ stays matched to $u'$). The process terminates when no universities are unfilled. 

\textbf{Stable Partners.} One can define, for all universities $u$, the set $S(u)$ of students that are ever matched to $u$ in any stable matching. These are called the stable partners of $u$. Similarly, one can define $U(s)$ to be the set of universities that are ever matched to $s$ in any stable matching. Much is known about the structure of stable partners:

\begin{theorem}\label{thm:GS}[\cite{GaleS62}]
The stable matching output at the end of the university-proposing deferred-acceptance algorithm is unique, and independent of the order in which universities propose (they may also propose simultaneously). Furthermore, for all universities $u$, $u$'s partners when the algorithm terminates is its $L$ favorite students in $S(u)$.
\end{theorem}

\subsection{Comparing Student and University Proposals}
In our model, once the set of applications have seen sent out, the schools form their (noisy) preferences among the received applications and applies the school-proposing deferred-acceptance algorithm. Deferred-acceptance is a reasonable approximation to the wait-list system used by many colleges. We can imagine a different, hypothetical way that the matching can be performed. Each student forms their preferences as described in our model. Then the student-proposing deferred acceptance is ran. Each student applies to her favorite school. The school either rejects the student---if it has already filled its slots with better applications as judged by the noisy signal---or tentatively accepts the student. Students who are rejected then apply to their next favorite schools, until each student makes up to $K$ applications. If a school receives a better application, it can reject a student that it had tentatively accepted before. We emphasize that this is just a thought-experiment model and is \emph{not} how school  application  works. In fact it sounds quite different from the school-proposing deferred-acceptance algorithm of our actual model. However,~\cite{ImmorlicaM15,KojimaP09} prove that in settings similar to ours, the resulting matchings are nearly identical. Using their techniques, we're able to show the same of our model as well:


\begin{theorem}\label{thm:equiv}
In our model, for constant $K, L, M$, and as $n \rightarrow \infty$, any two stables matchings are identical except for an $o(1)$ fraction of the students. In particular, the matchings produced by the school-proposing deferred-acceptance and the student-proposing deferred acceptance algorithms are the same except for an $o(1)$ fraction of the students. 
\end{theorem}

We defer the proof of Theorem~\ref{thm:equiv} to Section~\ref{sec:equiv_proof}. We shall analyze the student-proposing algorithm to derive an analytic characterization for the number of students that get into their top-choice schools and nontop-choice schools. Because Theorem~\ref{thm:equiv} guarantees that this outcome from student-proposing is essentially identical to the outcome from running our actual model using school-proposing, we know that the same expressions characterize the resulting matching from school-proposing deferred acceptance as well.

%% file: results.tex

\section{Solving our model}\label{sec:results}
In this section, we provide a closed form for computing the student and university utility as a function of parameters. To begin, consider the following thought experiment: Say we were to claim that as $n\rightarrow \infty$ that $x_i n \pm o(n)$ students wound up proposing to their $i^{th}$ favorite university in the student-optimal stable matching with probability $1-o(1)$. Then it better also be the case that $(1-x_i)n \pm o(n)$ students wound up matched to one of their top $i-1$ universities. This is simply because students who propose to their $i^{th}$ favorite university are exactly those who were rejected from all of their top $i-1$. 

Because of how student preferences are drawn in our model (namely, they are uniformly at random), this means that a total of $\sum_{i = 1}^{K} x_i n \pm o(n)$ proposals are made to uniformly random universities.\footnote{Technically, they are not quite uniformly at random, because the second proposal from a student is guaranteed not to go to her first choice. But since $K$ is fixed, as $n \rightarrow \infty$, this difference only impacts an $o(1)$ fraction of proposals.} Label all proposals from a student to their $i^{th}$ favorite university as a ``type $i$'' proposal. We know that in the student-proposing deferred acceptance, every university winds up matched to the best $L$ proposals they received during the run of the algorithm. Therefore, we could compute, given that $x_i n \pm o(n)$ type $i$ proposals are made, and given how university preferences are drawn over different kinds of proposals (depends only on $D, D'$), the expected number of type $i$ proposals that are \emph{accepted} in the end, for all $i$. We compute this for specific choices of $D, D'$ in Section~\ref{sec:applications}, but for now denote this as a a function parameterized by $D, D'$:

\begin{definition}
$f^{D, D',M,L}_i(\vec{x}) \cdot n$ denotes the expected number of type $i$ proposals that are accepted when $x_i n$ type $i$ proposals are made in total for $i \in [K]$, each to one of $Mn$ uniformly random universities, and each university selects their favorite $L$ proposals randomly according to preferences based on $D, D'$. Specifically, for each type $1$ proposal, the university will sample a signal from $D$. For each type $>1$ proposal, the university will sample a signal from $D'$. Then the university will select the $L$ students with the highest signals. 
\end{definition}

So now let's conclude our thought experiment: If indeed $x_2 n \pm o(n)$ students wound up proposing to their second favorite university, it better be the case that $x_2\cdot n = x_1\cdot n - f_1^{D, D',M,L}(\vec{x})\cdot n \pm o(n)$. Otherwise, there is an inconsistency in our claim: as the students whose top proposal is not accepted are exactly those who propose to their second choice. Similarly, if $x_3n \pm o(n)$ students wound up proposing to their third favorite university, it better be the case that $x_3 \cdot n = x_2\cdot n - f^{D, D',M,L}_2(\vec{x})\cdot n \pm o(n)$. The point is that once we know how many type $i$ proposals are made, for all $i$, we can turn around and compute the number of such proposals accepted (up to $o(n)$, with probability $1-o(1)$). But also, once we know the total number of proposals accepted, we can turn back around and compute how many proposals of each type were made (with probability $1-o(1)$). So if we are correct in guessing $\vec{x}$, we better get back $\vec{x} \pm o(n)$ after going back and forth. This provides a system of equations that $\vec{x}$ must satisfy if indeed we guessed correctly. This intuition is captured in Theorem~\ref{thm:main} below, whose proof consists mostly of formalizing the thought experiment above (which we therefore defer to Appendix~\ref{app:results}). In the statement below, note that $y_1 = 1$, as every student proposes to their top choice (for ease of notation, define $y_0 = 1$ as well, and $f_0^{D, D', M, L}(\vec{x}) = 0$ for all $\vec{x}, D, D', M,  L$). 

\begin{theorem}\label{thm:main}
Let $y_i \cdot n$ denote the expected number of students who propose to their $i^{th}$ favorite school in the student-optimal stable matching. Then: 
\begin{equation}\label{eq:main}
y_i = y_{i-1}- f_{i-1}^{D, D',M,L}(\vec{y}) \pm o(1), \ \forall i\in [K].
\end{equation}
\end{theorem}

To properly apply Theorem~\ref{thm:main}, note that it doesn't suffice to simply find a point in $[0,1]^K$ solving the system of equations, as there might be multiple. So in order to claim that a solution to the system of equations is in fact the expected number of type $i$ applications made for all $i$, we must also show that it is the unique solution. Fortunately, it is also the case that every solution implies the existence of a stable matching with roughly the same distribution of accepted proposals with probability $1-o(1)$. So if there were multiple solutions, this would imply that often there exist two stable matchings that differ on a $\Theta(n)$ fraction of students, which contradicts Theorem~\ref{thm:equiv}. We prove this below in Section~\ref{sec:unique}. 

\subsection{The Solution is Unique}\label{sec:unique}
In this section, we prove the following theorem:
\begin{theorem}\label{thm:unique}
Let $\vec{y}$ be a solution to the system of equations~\eqref{eq:main}. Then with probability $1-o(1)$ when preferences are drawn according to our model, there exists a stable matching where $(y_i-y_{i+1})n \pm o(n)$ students are matched to their $i^{th}$ favorite school for all $i \in [K-1]$.
\end{theorem}

The proof will again consist of a few steps. At a high level, our approach will be roughly to consider a student-proposing deferred acceptance, but starting with a bunch of proposals already having been made (namely, $y_i n - o(n)$ type $i$ proposals for all $i$), and show that not much additional proposing is necessary to yield a stable matching. Due to space constraints, we defer the proof to Appendix~\ref{app:unique}.

\begin{corollary}\label{cor:main}
In our model, the system of equations~\eqref{eq:main} has a unique solution in $[0,1]^K$. 
\end{corollary}
\begin{proof}
Assume for contradiction that there are multiple solutions $\vec{x}$ and $\vec{y}$. By Theorem~\ref{thm:unique}, a $1-o(1)$ fraction of instances both have a stable matching where $(y_{i+1}-y_i)n \pm o(n)$ students are matched to their $i^{th}$ favorite school for all $i \in [K-1]$, and where $(x_{i+1}-x_i)n \pm o(n)$ students are matched to their $i^{th}$ favorite school for all $i \in [K-1]$. In order to not contradict Theorem~\ref{thm:equiv}, it must be that $x_{i+1}-x_i = y_{i+1}-y_i$ for all $i \in [K-1]$. As $x_1 = y_1 = 1$, this immediately implies that $x_i = y_i$ for all $i$.
\end{proof}

%% file: applications.tex

\section{Utilities of matchings} \label{sec:applications}

In this section, we instantiate our results from the previous section for a couple choices of $D, D'$. In the first, we consider $D = D'$ (so the schools' preferences are also uniformly at random). In this model, if everyone were to apply everywhere, synergy would be completely hidden. So it is natural to expect that smaller $K$ yields higher expected utility for the schools (because synergy can only be detected by where students choose to apply). We also consider $D=\mathcal{N}(\delta, 1)$ and $D' = \mathcal{N}(0,1)$ for some $\delta > 0$. Here, if everyone were to apply everywhere, some synergy would be detected, but not perfectly. Here we see that depending on $\delta$, the optimal choice of $K$ could range from $1$ to $n$.

\subsection{Example: Synergy is completely hidden}
For this example, consider the case where $D = D'$. Then university preferences are also uniformly at random over the students. So every proposal is equally likely to be accepted, and the probability that a given proposal is accepted is just $\frac{\text{\# proposals accepted}}{\sum_i x_i}$. When $L = 1$, the number of proposals accepted is just the number of schools that receive a proposal, which is $Mn(1-e^{-\sum_i x_i/M})$ in expectation. When $L > 1$ the closed form for the expected number of accepted proposals can be written down explicitly but is messier, so we just denote it as $g(\sum_i x_i n, M, L)$. Theorem~\ref{thm:main} implies that if $y_i$ denotes the number of type $i$ proposals made, then we have:
$$y_j = y_{j-1} - \frac{g(\sum_i y_i n, M, L)}{\sum_i y_in}\cdot y_{j-1} \ \forall j \in [K].$$

Note that $\frac{g(\sum_i y_in, M, L)}{\sum_i y_i n}$ doesn't depend on $j$. So we can define $\alpha = 1-\frac{g(\sum_i y_i n, M, L)}{\sum_i y_i n}$, and recover that $y_j = \alpha^{j-1}$ for all $j$. So we just need to solve for $\alpha$. Note that we have two equations involving $\alpha$ and $\sum_i y_i$. The first is simply the definition of $\alpha$ we just computed. The second is that we know that for whatever $\alpha$ we wind up with, we must have $\sum_{j=1}^K \alpha^{j-1} = \sum_i y_i$ (as both sides compute the total number of applications made). Rewriting and substituting for $\alpha$, this implies:

\begin{equation*}
\sum_i y_i = \frac{1-(1-\frac{g(\sum_i y_i n, M, L)}{\sum_i y_i n})^K}{\frac{g(\sum_i y_i n, M, L)}{\sum_i y_i n}}.
\end{equation*}
\begin{equation}\label{eq:cool}
\Rightarrow g(\sum_i y_i n, M, L)/n = 1-(1-\frac{g(\sum_i y_i n, M, L)}{\sum_i y_i n})^K.
\end{equation}

It is easy to see that as $\sum_i y_i$ increases, the LHS of Equation~\eqref{eq:cool} increases (as more applications are made, the expected number of accepted applications increases). Similarly, it is easy to see that as $\sum_i y_i$ increases, the RHS of Equation~\eqref{eq:cool} decreases. As $\sum_i y_i$ increases, the fraction of applications that are rejected ($1-\frac{g(\sum_i y_i n,M,L)}{\sum_i y_i n}$) increases. So $(1-\frac{g(\sum_i y_i n,M,L)}{\sum_i y_i n})$ increases, and the RHS of Equation~\eqref{eq:cool} decreases. So if there is a solution it is unique (actually, we already know this by Corollary~\ref{cor:main}, but this is a good sanity check). We can also confirm that a solution exists by observing that as $\sum_i y_i \rightarrow \infty$, the LHS approaches $ML$, whereas the RHS approaches $0$. As $\sum_i y_i \rightarrow 0$, the LHS approaches $0$, whereas the RHS approaches $1$. So somewhere in between, a unique solution exists (again, we also knew this already from Corollary~\ref{cor:main}). In any case, we conclude the following:

\begin{proposition}\label{prop:iid}
When $D = D'$, the expected number of total applications sent in the student-proposing deferred-acceptance algorithm, $x$, is within $\pm o(n)$ of the unique solution to $g(x,M,L)/n = 1 - (1-g(x,M,L)/x)^K$. A $g(x, M, L)/x$ fraction of students get their first choice. 
\end{proposition}

Note that the information contained in Proposition~\ref{prop:iid} is indeed enough to compute the expected preference ranking of a university for its matched students. Also note that Equation~\eqref{eq:cool} helps quantify some intuition. Clearly, as $K$ increases, the total number of applications made in the student-proposing deferred acceptance will increase. We can readily read off from Equation~\eqref{eq:cool} that as $K$ increases, the RHS also increases, and therefore the total number of applications must also increase in order to re-reach equality. 

\subsection{Numerical results}
We performed numerical experiments to compute the number of students who get their first choice college as well as the total number of students that are matched with a college. We considered an instantiation of our general model  when there are 100 students and 100 colleges and each college has one slot. We analyze the effect as the students apply to an increasing number of colleges. The preferences of the students and of the colleges are sampled as described above. In particular, we think of the top-choice college of each student as being the ``right'' fit, so that both the student and the college receive increased utility when a student gets into her favorite college. 

\begin{figure}
  \centering
  \includegraphics[scale=0.6]{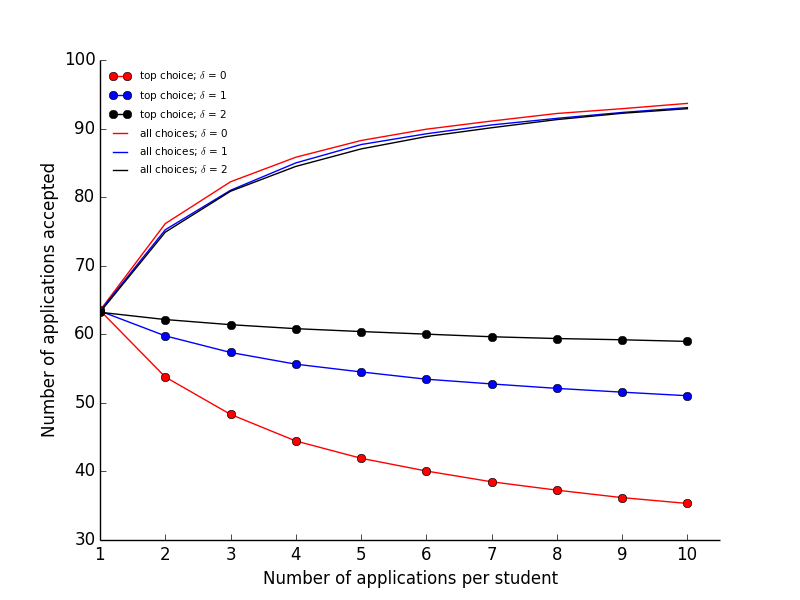}
  \caption{
As students make more applications, the number of students who get into their top-choice school decreases while the number of overall acceptances increases.}
  \label{fig:expt}
\end{figure}

The $x$-axis of Figure~\ref{fig:expt} indicates the number of applications made by each student. The bottom dotted curves correspond to the number of students that get into their top-choice school while the top curves correspond to the total number of students that get into a school. The different colors indicate different values of $\delta$. Recall that when $\delta = 0$, the applications are noisy and a college cannot distinguish between good-fit students and all other students. In this case, as the number applications made increases, a larger number of students are matched with a school, while a substantially smaller number of students get into their top choice. This confirms the theoretical analysis from the previous sections based on analyzing the fixed point equations. As $\delta$ increases, the ability of schools to identify the right students improves and we see less attribution of the top-choice applications. 

Depending on the ratio between the bonus utility from synergy versus the base utility from preferences, the optimal choice of $K$ will vary. As an example, if the base utility is $1$ for any match, and the bonus utility from synergy is also $1$, then for $\delta = 0$ in our experiments the utility is maximized when each student applies to 3 schools. 

%% file: equiv_proof.tex
\section{Discussion and Future Work}
We study the signal to noise tradeoff in matchings. We provide a clean and simple model for matching markets where one side has more information about match quality than the other, but there is no established signaling procedure for them to convey this information outside of the applications themselves. We show that our model is already rich enough to exhibit a tradeoff for both sides as the number of applications grows, without needing to include additional modeling features. Still, the parameters we choose to investigate in this work are exactly those for which a clean proof reveals the appropriate insight. It is an important direction for future work to investigate our model in regimes where these simple proofs no longer work. As an example, modeling synergy between students and schools besides their favorite would require tools beyond those developed in~\cite{ImmorlicaM15} to connect the student- and school-optimal stable matchings. As another example, modeling simultaneously the inefficiency caused by poor signaling \emph{and} application/processing cost would require additional reasoning about equilibria regarding how students choose which applications to send, and how schools decide which applications to process.

\section{Acknowledgements}
The authors are grateful to Peng Shi and Yannai Gonczarowski for guidance and helpful discussions during this work.

%% file: appendix.tex

\appendix

\section{Omitted proof of Theorem~2} \label{sec:equiv_proof}

We first recall a characterization of stable partners. 

\begin{theorem}\label{thm:Roth}[\cite{Roth84,McVitieW71}]
$\bot \in U(s) \Leftrightarrow \{\bot\} = U(s)$. That is, if a student is ever unmatched in a stable matching, it is unmatched in all stable matchings. Similarly, if $|S(u)| \leq L$, $u$ is matched to all of $S(u)$ in every stable matching. If $u$ is matched to $S$ with $|S|<L$ in some stable matching, then $S = S(u)$. 
\end{theorem}

\textbf{Determining Multiple Stable Partners.} The following theorem is reworded from~\cite{KojimaP09}, builds on work of~\cite{ImmorlicaM15,Knuth76,KnuthMP90}, and provides an algorithm to determine whether or not $|S(u)|>L$ for a given university $u$. Intuitively, the algorithm starts by finding $u$'s least favorite stable partners via the student-proposing deferred-acceptance, then does an exhaustive search for how $u$ might possibly wind up with a student they like better in a stable matching. Step 3-b-iii) bears witness that such a student exists. Otherwise, $u$'s only stable partners are those matched from the student-proposing deferred-acceptance.\footnote{Theorem~\ref{thm:KojimaP} is essentially Lemma 3 of~\cite{KojimaP09}, which is identical in technical content, but stated in the language of whether matchings can be profitably manipulated.}
\begin{theorem}[\cite{KojimaP09}]\label{thm:KojimaP} The following algorithm takes as input any university $u$, and determines whether or not $|S(u)|>L$.
\end{theorem}
\begin{enumerate}
\item Find the student-optimal stable matching. Let $M(u)$ denote the students matched to $u$. 
\item If $|M(u)| < L$, terminate and output \textsc{no} (because by Theorem~\ref{thm:Roth}, $M(u) = S(u)$). Otherwise:
\item For($S \subseteq M(u)$): For($s \in S$):
\begin{enumerate}
\item Unmatch $s$ and $u$, pretending that $u$ rejected $s$. 
\item While($s$ is unmatched and $s \neq \bot$):
\begin{enumerate}
\item $s$ proposes to their favorite university who hasn't yet rejected them, $u'$.
\item If $u' = \bot$: break, move onto the next $s \in S$ and start from 3a). 
\item If $u' = u$, and $u$ prefers $s$ to every student in $M(u)$, terminate and output \textsc{yes}.
\item If $u' \neq u$: update $s$ to be the student rejected by $u'$ (either the previous $s$, or their least favorite previous match, whichever is lower on $u'$'s preference list). 
\end{enumerate}
\end{enumerate}
\item If the algorithm completes the for-loops without outputting \textsc{yes}, output \textsc{no}. 
\end{enumerate}

\subsection{Number of Stable Partners}\label{sec:partners}
We show that as $n \rightarrow \infty$, a $1-o(1)$ fraction of universities have $|S(u)| \leq L$. This implies that the university-optimal stable matching and student-optimal stable matching are the same except for an $o(1)$ fraction of participants by Theorem~\ref{thm:Roth}. It turns out that it will be much simpler to analyze the student-optimal stable matching, so this is a useful observation. The proof is essentially identical to that of~\cite{ImmorlicaM15}, we just have to repeat it here for completeness since our model is slightly different.\footnote{Specifically: in their model it is assumed that university preferences are \emph{arbitrary}, but independent of student preferences. In our model, university preferences depend on student preferences, but not by much.} The proof is somewhat illustrative of the student-proposing deferred-acceptance process with random preferences. One takeaway from the proof is that mileage comes from the fact that a constant faction of universities will never get a proposal.

\begin{proposition}\label{prop:partners}
In our model, for constant $K, L, M$, and growing $n$, the expected number of universities with $|S(u)| > L$ is $o(n)$. 
\end{proposition}
\begin{proof}
We begin by applying the principle of deferred decisions to claim that the following is a valid way to draw student and university preferences in our model:
\begin{enumerate}
\item For each student independently, draw their favorite university uniformly at random, but not the rest of their preference list. 
\item For each university independently, draw their complete preferences.
\item For each student independently, complete their preferences, uniformly at random. 
\end{enumerate}

The point is that university preferences only depend slightly on student preferences: university $u$ is more likely to prefer special-$u$ students. But otherwise, university and student preferences are independent. In particular, university preferences are independent of student preferences with the exception of their favorite university. The rest of the proof continues exactly as in~\cite{ImmorlicaM15}. Pick a university $u$ and consider running the algorithm to determine whether or not $|S(u)|>L$, but using the principle of deferred decisions. Specifically, complete steps 1) and 2) in the above sampling, but not 3). Only draw the next university on a student's preference list when they are about to propose.

So first find the student-optimal stable matching, drawing all necessary student preferences in order to do so (and the complete university preferences). Certainly, each student will have proposed at least once at this point. Now, if $|M(u)| <L$, the algorithm terminates and correctly outputs \textsc{no}. Otherwise, the algorithm continues in step 3. In order to output \textsc{yes}, $s$ must propose to $u' = u$ (step 3-b-iii). However, if $s$ instead proposes to some university $u'$ who was unfilled in the student-optimal stable matching (or after the previous $s \in S$ were processed), then the algorithm returns to step 3-a) with a new $s$ having not yet output \textsc{yes} (because $u'$ will accept $s$'s proposal and the new $s$ will become $\bot$). If $s$ proposes to some university $u'\neq u$ who was matched in the student-optimal stable matching (of after the previous $s \in S$ were processed), then the algorithm continues with the same $s$, still having not yet output $\textsc{yes}$. Because each student has already proposed at least once to find the student-optimal stable matching, no matter which student $s$ is, or how many rejections they faced before being labeled $s$ in the algorithm, the university they will propose to next is currently undrawn, and is decided uniformly at random among all universities they haven't yet proposed to. So if $X$ denotes the number of universities who were unfilled in the student-optimal stable matching, then the probability that the algorithm ever reaches step 3-b-iii) (conditioned on $X$) is at most $\frac{2^LL}{X-L+1}$. This is because every time there is a chance of reaching step 3-b-iii), it is at least $X-L$ times as likely that $s$ proposes to an unfilled university (because up to an additional $L$ universities might become filled by the earlier $s \in S$) and the algorithm processes the next $s \in S$. So taking a union bound over all $s \in S$ that are processed, and over all $S$ that might be chosen results in $\frac{2^LL}{X-L+1}$. So now the analysis reduces to computing $\mathbb{E}[\frac{2^LL}{X-L+1}]$. 

Fortunately, this exact quantity is already bounded (Lemma 7/proof of Theorem 1) in~\cite{KojimaP09}. The idea is that for any university, the probability that it is unfilled in the student-optimal stable matching is at least the probability that it does not appear $L$ times in any student's top $K$ choices. There are a total of $nK$ universities among all students' top $K$ choices, all drawn independently with replacement. So the probability that a university doesn't appear $L$ times is at least the probability that a university doesn't appear at all in any student's top $K$ choices, which is exactly $(1-1/(nM))^{nK} \approx e^{-K/M}$, meaning that $\mathbb{E}[X] \approx nM/e^{-KM}$. Computing $\mathbb{E}[\frac{2^LL}{X-L+1}]$ is a touch more involved, but Lemma 7 (and the proof of Theorem 1) in~\cite{KojimaP09} show it to be $o(1)$ (intuitively, this should be true because $\mathbb{E}[X] = \Theta(n)$, and $L$ is constant). Combined with the analysis in the previous paragraph, this means that for a fixed $u$, the probability that $|S(u)| > L$ is at most $o(1)$. 
\end{proof}

\begin{proof}[Proof of Theorem~\ref{thm:equiv}]
Simply combine Proposition~\ref{prop:partners} and Theorem~\ref{thm:Roth}.
\end{proof}

\section{Omitted Proofs from Section~3}\label{app:results}
The proof of Theorem~\ref{thm:main} has a couple steps. We begin by showing that the number of type $i$ proposals accepted when $x_j n$ type $j$ proposals are made for all $j$ concentrates very tightly around $f_i^{D, D',M,L}(\vec{x})$.

\begin{lemma}\label{lem:accepted}
Let $x_j n$ type $j$ proposals to uniformly random universities be made for $j \in [K]$. Then with probability $1-o(1)$, for all $i$, the number of type $i$ proposals that are accepted is $f^{D, D',M,L}_i(\vec{x}) n \pm o(n)$. 
\end{lemma}

\begin{proof}
This is just a concentration bound. Let $X_p$ be the indicator random variable for the event that proposal $p$ is accepted. Then $\{X_p\}_{p\in P}$ are negatively correlated for all sets of proposals $P$ (because some proposals being accepted only make it less likely that others are). Therefore, the modified Chernoff-Hoeffding bound of~\cite{PanconesiS97,ImpagliazzoK10} applies to the random variable $\sum_{p|\text{ $p$ is type $i$}} X_p$. Therefore, this sum is within $\pm \sqrt{n\log n} = o(n)$ of its expectation with probability at least $1-2/n = 1-o(1)$. 
\end{proof}

\begin{lemma}\label{lem:accepted2}
Let $y_i \cdot n$ denote the expected number of students that propose to their $i^{th}$ favorite university in a run of the student-proposing deferred-acceptance. Then with probability $1-o(1)$, $y_i n \pm o(n)$ students propose to their $i^{th}$ favorite university in the student-optimal stable matching.
\end{lemma}
\begin{proof}
Again, this is just a concentration bound. Let $X_s$ be the indicator random variable for the event that student $s$ proposes to her $i^{th}$ favorite university. Then $\{X_s\}_{s \in S}$ are negatively correlated for all sets of students $S$ (because some students proposing a lot means they were rejected, making it more likely that other proposals were accepted, and therefore less likely that other students also propose to their $i^{th}$ favorite university). Therefore, the modified Chernoff-Hoeffding bound of~\cite{PanconesiS97,ImpagliazzoK10} again applies to the random variable $\sum_s X_s$, and the sum is within $\pm \sqrt{n\log n} = o(n)$ of its expectation with probability at least $1-2/n = 1-o(1)$. 
\end{proof}

\begin{proof}[Proof of Theorem~\ref{thm:main}]
Assume for contradiction that this was not the case. Lemma~\ref{lem:accepted2} shows that the number of students who propose to their $i^{th}$ favorite university is $y_i n \pm o(n)$ with probability $1-o(1)$. This means that the number of type $i$ proposals made is $y_i n \pm o(n)$ with probability $1-o(1)$. But Lemma~\ref{lem:accepted} then says that with probability $1-o(1)$, the number of accepted type $i$ proposals is $f_i^{D, D',M,L}(\vec{y})n \pm o(n)$. So $y_i$ better fall in the range $y_{i-1}n - f_{i-1}^{D, D', M, L}(\vec{y})n \pm o(n)$, or else we have a contradiction (that either too many or too few type $i$ proposals were made, given the number of rejected type $i-1$ proposals). 
\end{proof}

\section{Omitted Proofs from Section~3.1}\label{app:unique}We first consider drawing preferences in the following way:
\begin{enumerate}
\item For each $i$, draw $y_i \cdot n - o(n)$ random universities. These will be the type $i$ proposals made.
\item For each university $u$, draw a signal from $D$ for each type 1 proposal made to it, and a signal from $D'$ for each type $> 1$  proposal made to it. Select the $L$ proposals with highest signals. Label each proposal as ``accepted'' if it is one of the $L$ highest, or ``rejected'' otherwise.
\item Randomly assign all type $1$ proposals to the students. 
\item Repeat for $i = 2$ to $K$: Randomly assign the type $i$ proposals to students whose type $i-1$ proposal was rejected. Note that there might be fewer proposals than students. Leave the extra students without a type $i$ proposal.\footnote{Recall again that for a $\Theta(1/n)$ fraction of proposals, there will be an issue where a student's type $i$ proposal goes to the same school as their type $j$ proposal for $j \neq i$. This can again be safely ignored because all of our claims concern a $1-o(1)$ fraction of students, and again we will not address this formally.}
\item For all student preferences that have not yet been formed, drawn them uniformly at random.
\item To fill out the university preferences, if a signal for a proposal was already drawn, map that signal to the student who was assigned that proposal. Otherwise, draw a signal independently from $D'$ for all other proposals (which are necessarily type $> 1$), and sort students in decreasing order of signal. 
\end{enumerate}

Before appealing to the labels of ``accepted'' and ``rejected'' to construct a matching, let's first confirm that the above preferences are drawn faithfully w.r.t. our model.

\begin{proposition}\label{prop:sampling}
The sampling procedure above faithfully draws preferences in line with our model. 
\end{proposition}
\begin{proof}
Essentially we are again just using the principle of deferred decisions. From the university perspective, the point is that we don't need to know exactly which student made a proposal in order to rank it. We just need to know whether it was a type 1 or type $> 1$ proposal. So it is valid for universities to draw preferences over proposals without yet mapping those preferences to a concrete student. For the student preferences, we are simply drawing random universities to be in preference lists first (the list of ``made proposals''), and deferring the decision of exactly which student drew this school to be next on their list. But as all universities in the list are drawn uniformly at random, this is again a valid way to generate uniformly random preference lists for the students.
\end{proof}

Now, we want to claim that becuase $\vec{y}$ solves Equations~\eqref{eq:main}, we can recover an almost-stable matching from the above procedure. The idea is that if we match students to universities according to the ``accepted'' proposals, it would be a stable matching if only every student whose type $i$ proposal was labeled ``rejected'' had a type $i+1$ proposal assigned to them. Unfortunately, there are also $o(n)$ students whose type $i$ proposal was rejected, but who did not follow up with a type $i+1$ proposal. So these students will be unmatched, even though they didn't exhaust their preference list. We call such students inconsistent.

\begin{definition}
For preferences drawn from the above procedure, define a student to be \textbf{consistent} if for all $i$, they are assigned a type $i$ proposal if and only if they were assigned a ``rejected'' type $i-1$ proposal. Call all other students \textbf{inconsistent}. 
\end{definition}
\begin{proposition}\label{prop:matching}
Consider the matching that follows the ``accepted'' proposals. Specifically, if a student is assigned a proposal labeled ``accepted,'' they are matched to that school. If $\vec{y}$ solves Equations~\eqref{eq:main}, then with probability $1-o(1)$, the only blocking pairs for this matching contain inconsistent students.
\end{proposition}
\begin{proof}
First, we observe that with probability $1-o(1)$, there are always more rejected type $i$ proposals than type $i+1$ proposals made - this follows immediately from the fact that $\vec{y}$ solves Equations~\eqref{eq:main}, Lemma~\ref{lem:accepted}, and the fact that we made $y_in - o(n)$ type $i$ proposals (instead of exactly $y_i n$). So with probability $1-o(1)$, there is no issue with the same student getting mapped to two different ``accepted'' proposals, and the proposed matching is indeed a matching with probability $1-o(1)$. 

Now, consider any possible blocking pair $(s, u)$ where $s$ is consistent. Consistency implies that every school that $s$ prefers to its current match had $L$ ``accepted'' proposals that it prefered to $s$. Therefore, if $s$ prefers $u$ to its current match, $u$ does not prefer $s$ to any of its current matches, and $(s, u)$ can't possibly be a blocking pair. 

Note however that if $s$ is inconsistent, $(s, u)$ can be a blocking pair. If $s$ is inconsistent, then $s$ is unmatched. We know that for all the universities $u$ that $s$ ``proposed'' to, $u$ prefers all $L$ of its partners to $s$. But for the universities on $s$'s list that were generated in step 5) and not as ``proposals,'' we don't know whether or not these universities would prefer $s$ to their current partners, because the label of  ``accept'' and ``reject'' was only determined for the ``proposals.''
\end{proof}

At this point, we're very close to a stable matching, we just need to handle the $o(n)$ inconsistent students. To do this, we just need to let such students continue proposing down their list. This might cause a \emph{rejection chain}, where previously matched students get rejected and continue proposing, which causes other previously matched students to be rejected, etc. We show that the expected number of students affected by a single rejection chain is constant, and therefore all but $o(n)$ of the original matches remain. We begin by confirming that this procedure indeed results in a stable matching. The following lemma is well-known:

\begin{lemma}\label{lem:enter}
Let $\mathcal{M}$ be a stable matching for students $S$ and universities $U$. Then a stable matching for $S \cup \{s\}$ and $U$ can be found by beginning from $\mathcal{M}$, and continuing the student-proposing deferred-acceptance with $s$ proposing.
\end{lemma}
\begin{proof}
Let $\mathcal{M}'$ denote the matching at the end of the procedure, and consider any possible blocking pair $(t, u)$. If $t$ prefers $u$ to its current match, then either $t$ proposed to $u$ during the rejection chain, or $t$ was already matched to a university it liked worse than $u$ in $\mathcal{M}$. If $t$ proposed to $u$ during the rejection chain, then $t$ was rejected, so clearly $u$ has $L$ partners it prefers to $t$. If $t$ was already matched to a university it liked worse than $u$ in $\mathcal{M}$, then $u$ must have already had $L$ partners it prefered to $t$ in $\mathcal{M}$, as $\mathcal{M}$ was stable. As $u$'s partners can only improve going from $\mathcal{M}$ to $\mathcal{M}'$, $u$ must still have $L$ partners it prefers to $t$ in $\mathcal{M}'$, so $(t, u)$ cannot be a blocking pair.
\end{proof}

\begin{lemma}\label{lem:rejection}
Conditioned on steps 1-4 in the sampling process, and over the randomness in steps 5 and 6, the expected number of students affected by each rejection chain is $\leq 2e^{K/M}$.
\end{lemma}
\begin{proof}
Again using the principle of deferred decisions, we can imagine sampling a student's next favorite school only when necessary to continue the rejection chain. A rejection chain certainly terminates when a student's next favorite school has fewer than $L$ partners. If there are $X$ such universities, then the probability that a rejection chain terminates at each step is independently $X/(Mn)$. So the expected number of steps until the rejection chain terminates is just $Mn/X$. So we just need to compute $\mathbb{E}[Mn/X]$. The number of universities with fewer than $L$ partners is certainly at least the number of universities that isn't in any student's top $K$. The expected number of such universities is $Mn(1-1/(nM))^{nK} \approx Mne^{-K/M}$, and by Chernoff bound is at least $Mne^{-K/M}/2$ with probability $1-e^{-\Omega(n)}$. So $\mathbb{E}[Mn/X] \leq 2e^{K/M}$. 
\end{proof}

\begin{corollary}\label{cor:rejection}
Conditioned on steps 1-4 in the sampling process, and over the randomness in steps 5 and 6, the total number of students affected by each rejection chain is $o(n)$ with probability at least $1-o(1)$. 
\end{corollary}
\begin{proof}
This is just an application of Markov's inequality with appropriately chosen $o(1)$ to the random variable summing over all rejection chains of the number of students affected, which has expectation $o(n)$. 
\end{proof}

Now we are ready to conclude the proof of Theorem~\ref{thm:unique}:

\begin{proof}[Proof of Theorem~\ref{thm:unique}]
By Proposition~\ref{prop:sampling}, the sampling procedure provides a valid way to generate preferences in our model. By Proposition~\ref{prop:matching}, and the fact that $\vec{y}$ is solves Equations~\eqref{eq:main}, with probability $1-o(1)$, the resulting instance has a matching where $(y_{i+1}-y_{i})n \pm o(n)$ students are matched to their $i^{th}$ favorite school for all $i \in [K-1]$, and this matching is almost stable. By Lemma~\ref{lem:enter} and Corollary~\ref{cor:rejection}, the matching can be modified to a stable matching while only affecting $o(n)$ students. Therefore, the result is a stable matching where indeed $(y_{i+1}-y_i)n \pm o(n)$ students are matched to their $i^{th}$ favorite school for all $i \in [K-1]$. 
\end{proof}

%% file: MAIN.bbl
\newcommand{\etalchar}[1]{$^{#1}$}
\begin{thebibliography}{PRVW11}

\bibitem[AJK15]{ArnostiJK15}
Nick Arnosti, Ramesh Johari, and Yash Kanoria.
\newblock Managing congestion in dynamic matching markets.
\newblock {\em Manuscript}, 2015.

\bibitem[AN16]{AshlagiN16}
Itai Ashlagi and Afshin Nikzad.
\newblock What matters in tie-breaking rules? how competition guides design.
\newblock In {\em ACM Conference on Economics and Computation (EC)}, 2016.

\bibitem[ANR15]{AshlagiNR15}
Itai Ashlagi, Afshin Nikzad, and Assaf Romm.
\newblock Assigning more students to their top choices: A tiebreaking rule
  comparison.
\newblock In {\em ACM Conference on Economics and Computation (EC)}, 2015.

\bibitem[Arn16]{Arnosti16}
Nick Arnosti.
\newblock Centralized clearinghouse design: A quantity-quality tradeoff.
\newblock {\em Job Market Paper}, 2016.

\bibitem[AS14]{AshlagiS14}
Itai Ashlagi and Peng Shi.
\newblock Optimal allocation without money: An engineeringn approach.
\newblock {\em Management Science}, pages 1078--1097, 2014.

\bibitem[CCL{\etalchar{+}}10]{ColesCLNRS10}
Peter Coles, John Cawley, Phillip~B. Levine, Muriel Niederle, Alvin~E. Roth,
  and John~J. Siegfried.
\newblock The job market for new economists: A market design perspective.
\newblock {\em Journal of Economic Perspectives}, 24(4):187--206, 2010.

\bibitem[CKng]{CheK15}
Yeon-Koo Che and Youngwoo Koh.
\newblock Decentalized college admissions.
\newblock {\em Journal of Political Economy}, Forthcoming.

\bibitem[GF13]{GonczarowskiF13}
Yannai Gonczarowski and Ehud Friedgut.
\newblock Sisterhood in the gale-shapley matching algorithm.
\newblock {\em The Electronic Journal of Combinatorics}, 20(2), 2013.

\bibitem[GI89]{GusfieldI89}
Dan Gusfield and Robert~W. Irving.
\newblock {\em The stable marriage problem: structure and algorithms}.
\newblock MIT Press, 1989.

\bibitem[GS62]{GaleS62}
D.~Gale and Lloyd~S. Shapley.
\newblock College admissions and the stability of marriage.
\newblock {\em American Mathematical Monthly}, 69:9--14, 1962.

\bibitem[GS85]{GaleS85}
David Gale and Marilda Sotomayor.
\newblock Some remarks on the stable matching problem.
\newblock {\em Discrete Applied Mathematics}, 11(3):223--232, 1985.

\bibitem[HPY16]{HalaburdaPPY16}
Hanna Halaburda, Mikolaj~Jan Piskorski, and Pinar Yildrim.
\newblock Competing by restricting choice: The case of search platforms.
\newblock {\em Harvard Business School Strategy Unit Working Paper No. 10-098},
  2016.

\bibitem[IAL]{AshlagiKL16}
Yash~Kanoria Itai~Ashlagi and Jacob Leshno.
\newblock Unbalanced random matching markets: the stark effect of competition.
\newblock {\em Journal of Political Economy (forthcoming)}.

\bibitem[IK10]{ImpagliazzoK10}
Russell Impagliazzo and Valentine Kabanets.
\newblock Constructive proofs of concentration bounds.
\newblock In Maria~J. Serna, Ronen Shaltiel, Klaus Jansen, and Jos{\'{e}} D.~P.
  Rolim, editors, {\em Approximation, Randomization, and Combinatorial
  Optimization. Algorithms and Techniques, 13th International Workshop,
  {APPROX} 2010, and 14th International Workshop, {RANDOM} 2010, Barcelona,
  Spain, September 1-3, 2010. Proceedings}, volume 6302 of {\em Lecture Notes
  in Computer Science}, pages 617--631. Springer, 2010.

\bibitem[IM15]{ImmorlicaM15}
Nicole Immorlica and Mohammad Mahdian.
\newblock Incentives in large random two-sided markets.
\newblock {\em Transactions on Economics and Computation}, 2015.

\bibitem[Kim16]{Kim16}
Anne Kim.
\newblock How the internet wrecked college admissions.
\newblock {\em Washington Monthly}, 2016.
\newblock
  http://washingtonmonthly.com/magazine/septemberoctober-2016/how-the-internet-wrecked-college-admissions/.

\bibitem[KM09]{KobayashiM09}
H.~Kobayashi and T.~Matsui.
\newblock Successful manipulation in stable marriage model with complete
  preference lists.
\newblock {\em IEIC Transactions on Information and Systems}, 2009.

\bibitem[KMP90]{KnuthMP90}
Don~E. Knuth, Rajeev Motwani, and Boris Pittel.
\newblock Stable husbands.
\newblock In {\em Proceedings of the First Annual ACM-SIAM Symposium on
  Discrete Algorithms}, 1990.

\bibitem[Knu76]{Knuth76}
D.E. Knuth.
\newblock Mariages stables et leurs relations avec d'autres probl\`{e}mes
  combinatoires.
\newblock {\em Montr\'{e}al: Les presses de l'Universit\'{e} de Montr\'{e}al},
  1976.

\bibitem[KP09]{KojimaP09}
Fujito Kojima and Parag~A. Pathak.
\newblock Incentives and stability in large two-sided matching markets.
\newblock {\em American Economic Review}, 2009.

\bibitem[LN15]{LeeN15}
Soohyung Lee and Muriel Niederle.
\newblock Propose with a rose? signaling in internet dating markets.
\newblock {\em Experimental Economics}, 18(4):731--755, 2015.

\bibitem[LNKK11]{LeeNKK11}
Soohyung Lee, Muriel Niederle, Hye-Rim Kim, and Woo-Keum Kim.
\newblock Propose with a rose? signaling in internet dating markets.
\newblock {\em National Bureau of Economic Research Working Paper 17340}, 2011.

\bibitem[MW71]{McVitieW71}
D.~G. McVitie and L.~B. Wilson.
\newblock The stable marriage problem.
\newblock {\em Communications of the ACM}, 14(7), 1971.

\bibitem[Pit89]{Pittel89}
Boris Pittel.
\newblock The average number of stable matchings.
\newblock {\em SIAM Journal on Discrete Mathematics}, 2(4):530--549, 1989.

\bibitem[PRVW11]{PiniRVW11}
M.S. Pini, F.~Rossi, K.~B. Venable, and T.~Walsh.
\newblock Manipulation complexity and gender neutrality in stable marriage
  procedures.
\newblock {\em Autonomous Agents and Multi-Agent Systems}, 22(1):183--199,
  2011.

\bibitem[PS97]{PanconesiS97}
Alessandro Panconesi and Aravind Srinivasan.
\newblock Randomized distributed edge coloring via an extension of the
  chernoff-hoeffding bounds.
\newblock {\em {SIAM} J. Comput.}, 26(2):350--368, 1997.

\bibitem[Rot82]{Roth82}
Alvin~E. Roth.
\newblock The economics of matching: stability and incentives.
\newblock {\em Mathematics of Operations Research}, 7(4):617--628, 1982.

\bibitem[Rot84]{Roth84}
Alvin~E. Roth.
\newblock The evolution of the labor market for medical interns and residents:
  A case study in game theory.
\newblock {\em Journal of Political Economy}, 92:991--1016, 1984.

\bibitem[RX97]{RothX97}
Alvin~E. Roth and Xiaolin Xing.
\newblock Turnaround time and bottlenecks in market clearing: Decentralized
  matching in the market for clinical psychologists.
\newblock {\em Journal of Political Economy}, 105:184--329, 1997.

\end{thebibliography}
